\documentclass{article}

\usepackage{pdfpages}
\usepackage{mathtools}
\usepackage{amsfonts}
\usepackage{amsthm}
\usepackage{algorithm,algorithmic}

\usepackage{wrapfig}
\usepackage{blkarray}
\usepackage{dsfont}
\usepackage{balance}
\usepackage{fullpage}
\usepackage{url, hyperref}

\usepackage{cleveref}

\newcommand{\ind}{\mathds{1}}
\DeclareMathOperator*{\E}{\mathbb{E}}
\newcommand{\tr}{\operatorname{tr}}
\newcommand{\ospace}{\mathcal{O}}
\newcommand{\ispace}{\mathcal{I}}
\newcommand{\pmat}{\mathbf{P}}
\newcommand{\pmatdist}{\mathcal{D}}
\newcommand{\pmatdistfamily}{\mathbb{D}}

\newcommand{\partialinfo}{W}
\newcommand{\topics}{\mathcal{T}}

\newcommand{\amat}{\mathbf{A}}

\newcommand{\Rset}{\mathbb{R}}
\renewcommand{\Pr}{\mathbb{P}}
\newcommand \reals {\Rset}
\newcommand \acc {\operatorname{Acc}}
\newcommand \accr {\acc_\text{R}}
\newcommand \accm {\acc_\text{M}}

\newtheorem{example}{Example}
\newtheorem{thm}{Theorem}
\def\E{\mathbb{E}}
\def\ind{\mathbb{I}}

\newtheorem{lemma}[]{Lemma}
\newtheorem{proposition}[]{Proposition}
\newtheorem{corollary}{Corollary}
\newtheorem{assumption}{Assumption}
\theoremstyle{definition}
\newtheorem{definition}[]{Definition}

\newcommand \email [1]{\small{\url{#1}}}

\title{Measuring Re-identification Risk}

\begin{document}

\author{
CJ Carey \\ Google \\ \email{cjcarey@google.com}
\and
Travis Dick \\ Google \\ \email{tdick@google.com}
\and
Alessandro Epasto \\ Google \\ \email{aepasto@google.com}
\and
Adel Javanmard \\ USC\footnote{Also affiliated with Data Sciences and Operations, University of Southern California.}, Google \\ \email{adeljavanmard@google.com}
\and
Josh Karlin \\ Google \\ \email{jkarlin@google.com}
\and
Shankar Kumar \\ Google \\ \email{shankarkumar@google.com}
\and
Andr\'es Mu\~noz Medina \\ Google \\ \email{ammedina@google.com}
\and
Vahab Mirrokni \\ Google \\ \email{mirrokni@google.com}
\and
Gabriel Henrique Nunes \\ UFMG\footnote{Also affiliated with Universidade Federal de Minas Gerais.}, Google \\ \email{ghn@nunesgh.com}
\and
Sergei Vassilvitskii \\ Google \\ \email{sergeiv@google.com}
\and
Peilin Zhong \\ Google \\ \email{peilinz@google.com}
}

\maketitle

\begin{abstract}
Compact user representations (such as embeddings) form the backbone of personalization services. In this work, we present a new theoretical framework to measure re-identification risk in such user representations. Our framework, based on hypothesis testing, formally bounds the probability that an attacker may be able to obtain the identity of a user from their representation. As an application, we show how our framework is general enough to model important real-world applications such as the Chrome's Topics API for interest-based advertising. We complement our theoretical bounds by showing provably good attack algorithms for re-identification that we use to estimate the re-identification risk in the Topics API. We believe this work provides a rigorous and interpretable notion of re-identification risk and a framework to measure it that can be used to inform real-world applications. 
\end{abstract}

\section{Introduction}
 From curated travel suggestions to localized search results and relevant ads, personalization of online content has become increasingly important, as users expect online systems to be intelligent enough to anticipate their needs. The majority of these systems are now powered by so-called user representations --- for instance dense vector embeddings in $\Rset^d$ or discrete tokens---that enable classifiers and recommendation systems to generate useful personalized content. %

User representations form a compact description of a user profile that distill a user's interest into a short vector. For instance, a music streaming service may represent a user by a genre of songs the user likes. More sophisticated embeddings may be trained using neural networks and not be as readily explainable. 

The compact representation of user profiles may be seen as providing some privacy: instead of describing a user by all raw data (e.g. all songs they have listened to), one can summarize their profile in a few bits. Of course, such informal privacy arguments can be incorrect, and do not account for sophisticated attacks or non-obvious data leak vectors. In this work we study the question formally. Specifically, we ask to what extent user representations can be used to link back to and re-identify individuals. We provide a rigorous and interpretable notion of re-identification risk and a framework to measure it.

\paragraph{Assessing Privacy}
Since the seminal work of Dwork et al.~\cite{dwork2006calibrating} introducing differential privacy, significant effort has been devoted into developing differentially private algorithms for numerous problems and settings. Despite these efforts, far from all production systems in use currently employ differential privacy. Moreover, the ever evolving legal framework for privacy (such as GDPR~\cite{cohen2020towards} and other regulations) has introduced over time even more privacy definitions, while research has pointed out that these definitions are often not compatible with basic privacy properties like composability~\cite{cohen2020towards}. 

An immediate question is how to assess the privacy of such real-world systems. If the underlying method is randomized, one can try to establish its sensitivity to a single user's input and derive differential privacy bounds. However, just as in the GDPR example highlighted by~\cite{cohen2020towards}, the bounds may be vacuous. In this work we propose a complementary approach. We develop a rigorous new framework, centered around re-identification risk, and prove two kinds of results. First, is a series of unconditional lower bounds, showing the risk present in the system. Second, we give a formalization of upper bounds under a ``closed-world'' assumption, where we can characterize all of the  information an attacker has access to. The closed world assumption is natural in some scenarios, and unrealistic in others, in either case our framework can provide a spectrum of the re-identification risk incurred by the users.   

We remark that measuring re-identification risk complements the guarantees provided by other notion of privacy, such as (local) differential privacy and k-anonymity, by giving guarantees in a different domain. As part of our analysis, we derive re-id risk guarantees for differentially private algorithms, and k-anonymous datasets. Critically, however, this notion can provide a measure of risk for systems adopting other ad-hoc privacy approaches, an important contribution given the large number of products not designed with differential privacy guarantees.

\paragraph{Applications}
One concrete example we study is the the Topics API~\cite{topics-github} proposed by Chrome as part of the Privacy Sandbox initiative~\cite{privacysandbox}.  In a nutshell, the Topics API creates a representation of a user corresponding to the top interests of the user in a taxonomy of about 350 interests. When a website calls the Topics API, it returns one of the top interests of the user uniformly at random (and with some small probability a random topic from the taxonomy not necessarily in the user representation). We present a formal study quantifying the extent to which the Topics API can be used to re-identify users across different domains. 

We couple this work with an experiment quantifying the risk in releasing samples  of music listened by users in a large song dataset. Using our framework, we can quantify re-identification risk in this example, and show that with 4 independent samples per user, the re-id risk is around 1\%.

\subsection{Main contributions}
In summary, we make the following main contributions:
\begin{itemize}
\item We describe a hypothesis testing~\cite{blahut1974hypothesis} framework for evaluating the risk of identifying a random user given access to a representation as well as knowledge of the (potentially randomized) process that generates the representation. 
\item We extend this scenario and  consider a case where an attacker not only observes the representation of a random user but also observes representations of all users and uses this information to match representations to user identities. We refer to this scenario as the matching setting.
\item We measure the re-identification risk against an attacker that does not fully know the representation generation distribution (e.g., due to uncertainty on the underlying user data) but has a prior over the set of distributions that generated  the representations. 
\item We specialize this scenario to the Topics API for which we provide an in-depth analysis of an attacker's re-identification ability under this model. 
\item  We validate our results on two datasets. First a publicly available song dataset, and second a simulation of the Topics API based on proprietary data.  %
\end{itemize}

\section{Related Work}

Our work builds upon the rich literature on data anonymization~\cite{anonymization}, statistical privacy~\cite{statistical_privacy}, representation learning~\cite{representation_learning}, and web fingerprinting and tracking~\cite{browser_fingerprinting} and more generally on privacy preserving data mining and analysis~\cite{agrawal2000privacy,evfimievski2003limiting,verykios2004state}.

Here we discuss only elements of these areas that most relate to the problem of user re-identification. In Section~\ref{sec:comparison} we provide a more in-depth study of how our notion of re-identification relates to other privacy definitions.

\paragraph{Privacy preserving data releases}
This privacy framework traditionally models a data curator that wants to release data without leaking individual user information. In the context of sharing user representations, two of the most popular schemes for doing this are $k$-anonymity~\cite{kanon} and differential privacy~\cite{dwork}. In $k$-anonymity, the data curator ensures that each user representation is shared by at least another $k$ users. This is normally achieved by using optimization techniques for minimizing error in redacting tabular data~\cite{kenig2012practical,park2010approximate,ilprints645,meyerson2004complexity}, as well as  clustering algorithms~\cite{floc}, and tree-based methods~\cite{lefevre2006mondrian}. Differential privacy and in particular local differential privacy \cite{localdp}, introduce noise into the user representations to bound that information about a user that is leaked by the representation. There is a vast literature on differential privacy~\cite{dwork2006calibrating,mcsherry2007mechanism,chaudhuri2011differentially,roy2020crypte,machanavajjhala2017differential,dwork2019differential} (we refer to~\cite{dwork} for a survey). Significant work in differential privacy has been devoted to private representation learning and machine learning~\cite{abadi2016deep}.

In Section~\ref{sec:comparison}, we show that these two notions are theoretically sufficient (but not necessary) to prove low re-identification risk.

\paragraph{Browser fingerprinting and web tracking}
It is well known that data brokers and ad tech providers can use technologies such as third-party cookies to build detailed browsing history profiles of users~\cite{3pcookietracking}. With the impending deprecation of third-party cookies by major browsers, significant attention has been focused by the research community on understanding covert ways of tracking. One such method is browser fingerprinting~\cite{browser_fingerprinting}, which relies on characteristics of devices including the screen size and operating system versions to uniquely identify a user across the web. Several studies~\cite{unique_browser,fingerprintingfr,ml_for_fingerprinting} have analyzed how much information (measured by entropy) trackers can gain from such APIs to create a user identifier. The more entropy an API has, the more likely is that it can be used to re-identifying a user. With some few exceptions~\cite{battery_fingerprinting} most of the research work in this area has focused on understanding the re-identification risk of an API over a fixed instant in time. Our work expands this area of research by modeling the information leakage across time as values returned by an API can change (see section~\ref{sec:case-study}). Moreover, we believe our notion of re-identification can have a more straight-forward semantic interpretation compared to the number of bits leaked by an API as it directly measures the re-identification risk. 

\paragraph{Information theoretic notions of privacy}
Finally, there is a rich body of literature \cite{daleniousprivacy,miprivacy,robustinformationtheory} devoted to studying privacy leakage of information release through the lens of channel capacity. In the context of this paper, a channel is an object that takes as input user information and outputs a user representation. The information theoretic view of privacy measures the ability of an attacker with access to the output of a channel (and with knowledge of the channel internal mechanism) to reconstruct the input to the channel. The quantification of this ability traditionally involves analyzing the mutual information~\cite{elementsofit} between the input and output of the channel and its generalizations~\cite{renyimi}. In  Section~\ref{sec:comparison} we show how one can derive bounds on re-identification risk using mutual information.  Another related work is that of Cohen et al.~\cite{cohen2020towards} that formalizes privacy as a protection against singling out a user. Unlike our paper, it makes certain assumptions on data generation not present in our work (e.g., that data must be i.i.d. from certain distributions).

Finally, we compare our framework to that of  quantitative information flow (QIF)~\cite{qif}, a recent generalization of information theory which provides a flexible and semantically sound framework for analyzing the privacy risks of a channel. As we shall see later on, elements of our framework can be seen as a special case of the QIF framework. By specializing it however, we are able to provide tight characterizations of optimal re-identification attacks.

\section{Model}\label{sec:model}
We consider a universe of $n$ users indexed by identities in the set $\ispace = [n]$ and an attacker whose goal is to re-identify users based on their representations. Our goal is to quantify the extent to which an attacker can re-identify users as a function of how the user representations are constructed. 
We study this problem in two settings: First, in the random-user setting, one of the $n$ users is chosen uniformly at random from $\ispace$, that user's representation is revealed to the attacker, and the attacker attempts to guess the user's identity. 
Next, in the matching setting, the attacker observes representations for all $n$ users, and their goal is to match each of the $n$ representations to one of the $n$ user identities. Notice how this modeling can be seen as a form of hypothesis testing, where the hypothesis corresponds to the possible identities of the user.

We begin by formally defining the processes that generate representations for each user. We then describe the formal attack model, which quantifies how attackers interact with these representations.

\paragraph{Representations}
We assume that user representations belong to a representation space $\ospace$ with finite cardinality $|\ospace| = m$. For instance, $\ospace$ may be the set of all music genres.\footnote{We make the finite cardinality assumption for simplicity, as it is sufficient to elucidate the main results of the paper and to model the Topics API. We believe however our theory can be extended to dense embeddings as well.} Note that we do not make any further structural assumptions on $\ospace$. For example, $\ospace$ may be a finite set of topics, songs, images, or points in $\reals^d$.

An important part of our setup is that a user is not associated with a single representation, but rather with a {\em distribution} over representations. In other words, the process that assigns a user to a representation may be randomized. As a running example, suppose a user is equally interested in {\em classical} music and {\em alt-rock}. One way to represent this behavior is as a fractional assignment $\{0.5, 0.5\}$ to the two genres. Then, whenever a single genre is required, one can select it from the given distribution. (Looking ahead, this fractional assignment will be key for strengthening  re-identification protections.) 

Formally, we can encode the representation distributions assigned to all $n$ users using a row-stochastic \emph{representation matrix} $\pmat \in [0,1]^{n \times m}$ where $\pmat[i,o]$ is the probability that user $i \in \ispace$ has representation $o \in \ospace$  (where we slightly abuse notation and use elements of $\ospace$ to index columns of $\pmat$).
We write $\pmat[i, :]$ to denote the $i^{\rm th}$ row of $\pmat$, which is the user $i$'s distribution over $\ospace$.\footnote{We stress that the set $O$ is arbitrary and this allows representing arbitrary discrete distributions, including over high dimensional spaces. We leave generalizing our framework to continuous distributions as a future work.}

Finally, in order to reason about potentially randomized algorithms that assign representations to users, we model the process that constructs the representation matrix $\pmat$ (i.e., that assigns representations to users) as a distribution $\pmatdist$ over representation matrices in $[0,1]^{n \times m}$. The distribution $\pmatdist$ models both the user behavior informing the representations and algorithms used to construct the assignment of representation distributions to users.

\paragraph{Attack Model}
We begin with nature sampling a single representation matrix $\pmat$ from $\pmatdist$. At that point, attackers will attempt to re-identify users based on representations sampled from $\pmat$. We detail this further below.  %

The attacker's goal is to re-identify users based on their representations  (instead of using them for the intended use case---e.g., personalization of content). We categorize attackers by the varying degrees of knowledge about the sampled representation matrix $\pmat$. 

In the {\em full-information} setting, we assume that the attacker observes the sampled representation matrix $\pmat$. 
This corresponds to a powerful attacker with detailed knowledge of all users and the representation generation process $\pmatdist$.
In the partial-information setting, we assume a weaker attacker: one that receives a vector containing one representation sampled from each user's representation distribution: $\partialinfo \in \ospace^n$ where $\partialinfo_i \sim \pmat[i,:]$. %
This corresponds to a situation where an attacker only learns about $\pmat$ as a client consuming representations. 

In both scenarios, the attacker uses their knowledge of $\pmat$ (either the actual $\pmat$ itself or the vector of representations $\partialinfo$) in order to construct a prediction rule $\varphi$ that they will use to re-identify or match users. In the full-information setting, we require that $\varphi$ be $\pmat$-measurable, while in the partial-information setting, $\varphi$ must be $\partialinfo$-measurable.

\paragraph{Success Metrics}
In the random-user setting, the attacker attempts to re-identify a single randomly chosen user based on a sample of their representation. 

Formally, the attacker uses their knowledge of $\pmat$ to construct a possibly randomized prediction rule $\varphi_R : \ospace \to \ispace$. We define the accuracy random variable as follows: let $I$ be a uniformly random sample from $[n]$,  $O$ be sampled from $\pmat[I, :]$, and define
\[
\accr(\varphi_R) = \Pr(\varphi_R(O) = I \,|\, \pmat),
\]
which is the probability that the attacker correctly re-identifies the random user conditioned on the representation matrix $\pmat$.

In the matching setting, the attacker receives a set of representations, one for each user, and their goal is to match them to the user identities. 

Formally, the attacker uses their knowledge of $\pmat$ to construct a matching rule $\varphi_M : \ospace^n \to \ispace^n$. We define the matching accuracy random variable as follows: Sample a permutation $\pi : [n] \to [n]$ uniformly at random, a vector of independent representations $O_{1:n} \in \ospace^n$, where $O_i \sim \pmat[\pi(i), :]$ is a representation sampled for user $\pi(i)$, and define
\[
\accm(\varphi_M) = \E\left[\frac{1}{n}\sum_{i=1}^n \ind\{\varphi_M^i(O_{1:n}) = \pi(i) \} \,\bigg|\, \pmat\right],
\]
where $\varphi_M^i(O_{1:n})$ denotes the $i^\text{th}$ component of $\varphi_M(O_{1:n})$, which is the attacker's prediction for the user identity that produced observation $O_i$.
This is the expected fraction of users the attacker correctly re-identifies, conditioned on the representation matrix $\pmat$. 

We do not require that $\varphi_M(O_{1:n})$ is a permutation of $\ispace$. Note that the random permutation $\pi$ is used to ensure that the indices in the representation vector $O_{1:n}$ are not useful for re-identification. In particular, any constant prediction rule $\varphi_M$ has $\accm(\varphi_M) = 1/n$.

\section{Random-user Accuracy Bounds} \label{sec:randomUser}

In this section we provide information theoretic bounds on the accuracy that any attacker can achieve in the random-user setting. 
Recall that, in the random user setting, we sample a representation matrix $\pmat$ from the distribution $\pmatdist$ and the attacker formulates a (possibly randomized) prediction rule $\varphi_R : \ospace \to \ispace$ based on their knowledge of $\pmat$.
We let $I$ be a user drawn uniformly at random from $\ispace$ and $O \in \ospace$ be a representation sampled from $\pmat[I,:]$, and the accuracy random variable, $\accr(\varphi_R)$, is the probability that $\varphi(O) = I$, conditioned on the representation matrix $\pmat$.
Our main result  holds with probability one over the draw of the representation matrix $\pmat$ from $\pmatdist$. 
First, for any prediction rule $\varphi_R : \ospace \to \ispace$, we provide an exact expression for $\accr(\varphi_R)$. 
Next, we prove an upper bound on $\accr(\varphi_R)$ that depends only on $\pmat$. 
Our upper bound is tight in the full information setting.

Before stating our results, we introduce a matrix representation for any (possibly randomized) prediction rule $\varphi_R : \ospace \to \ispace$.
Define the matrix $\amat \in [0,1]^{n \times m}$ to have entries $\amat[i,o] = \Pr(\varphi_R(o) = i \,|\, \pmat)$. 
That is, conditioned on $\pmat$, $\amat[i,o]$ is the probability that the attacker's rule predicts user $i \in \ispace$ after receiving representation $o \in \ospace$.
With this, we are ready to state our main results.

\begin{lemma}\label{randomUserAccuracy}
    Let $\pmat \sim \pmatdist$ be a random representation matrix and $\amat$ be the matrix representation of an attacker's prediction rule $\varphi_R : \ospace \to \ispace$. Then, with probability one, we have that
    \[
    \accr(\varphi_R) = \frac{1}{n}\tr(\pmat\amat^\top).
    \]
\end{lemma}
\begin{proof}
    Let $I$ be a user chosen uniformly at random from $\ispace$ and $O \in \ospace$ be a representation sampled from $\pmat[I,:]$. 
    From the law of total probability, we have that
    \begin{align*}
        \Pr(\varphi(O) = I \,|\, \pmat)
        = &\sum_{o \in \ospace} \sum_{i=1}^n
           \Pr(\varphi_R(O) = I \,|\, O=o, I=i, \pmat)
           \cdot \Pr(O=o, I=i \,|\, \pmat) \\
        =& \frac{1}{n}\sum_{o \in \ospace} \sum_{i=1}^n 
           \amat[i,o] \cdot \pmat[i,o],
    \end{align*}
    where the double sum is equal to the Frobenius inner product of $\pmat$ and $\amat$, which can also be written as $\tr(\pmat \amat^\top)$.
\end{proof}

Next, we provide upper bounds on $\accr(\varphi_R)$ in terms of the representation matrix $\pmat$.

\begin{corollary}\label{randomUserAccuracyBound}
    Let $\pmat \sim \pmatdist$ be a random representation matrix and let $\varphi_R$ be any prediction rule. 
    Then, with probability one, we have that
    \[
    \accr(\varphi_R) \leq \frac{1}{n}\sum_{o \in \ospace} \max_{i \in \ispace} P[i,o] =:\frac{1}{n} \|\pmat\|_{\infty,1}.
    \]
    Additionally, there exists an attacker in the full-information setting capable of constructing $\varphi_R^*$ such that $\accr(\varphi_R^*) = \frac{1}{n} \|\pmat\|_{\infty,1}$.
\end{corollary}
\begin{proof}
    Let $\varphi_R : \ospace \to \ispace$ be any attacker prediction rule and let $\amat \in [0,1]^{n \times m}$ be its matrix representation. 
    By \Cref{randomUserAccuracy}, the following holds with probability one:
    \begin{flalign*}
        \accr(\varphi_R)& 
        = \frac{1}{n} \sum_{o \in \ospace} \sum_{i=1}^n \amat[i,o] \pmat[i,o] 
        \leq \frac{1}{n} \sum_{o \in \ospace} \sum_{i = 1}^n
          \amat[i,o] \cdot \max_{j \in [n]} \pmat[j,o] \\
        &= \frac{1}{n} \sum_{o \in \ospace} \max_{i \in [n]} \pmat[i,o] 
        = \frac{1}{n} \|\pmat\|_{\infty,1},
    \end{flalign*}
    where the second last equality follows from the fact that for any $o \in \ospace$, we have $\sum_{i=1}^n \amat[i,o] = 1$.
    
    Finally, the inequality in the above derivation holds with equality whenever, for each representation $o \in \ospace$ and user $i \in \ispace$, $\amat[i,o] > 0$ implies that $\pmat[i,o] = \max_j \pmat[j,o]$. 
    An attacker in the full-information setting can design $\amat$ so that this holds, which implies they are able to achieve this accuracy exactly. 
    In other words, if the attacker designs $\varphi_R^*$ so that $\varphi_R^*(o)$ is always some user $i$ with the maximum probability of generating representation $o$, then $\accr(\varphi_R) = \frac{1}{n}\|\pmat\|_{\infty,1}$.
\end{proof}

Next, we prove a bound on the accuracy of an attacker with partial-information. In particular, we bound their expected accuracy conditioned on the partial information contained in $W \in \ospace^n$. This is the probability that the attacker will correctly predict the identity of a random user based on a sample representation, conditioned on vector $\partialinfo$.

\begin{lemma}\label{lemma:partial_info_bound}
    Let $\pmat \sim \pmatdist$ be a random representation matrix, $\partialinfo \in \ospace^n$ be a vector of independent representations where $\partialinfo_i \sim \pmat[i,:]$, and let $\varphi_R : \ospace \to \ispace$ be a prediction rule that is $W$-measurable. 
    Then we have that
    \[
    \E\bigl[ \accr(\varphi_R) \,|\, W \bigr] \leq \frac{1}{n} \bigl\|\E[\pmat \,|\, W]\bigr\|_{\infty, 1}.
    \]
    Moreover there exists a prediction rule $\varphi_R^*$ for which the upper bound is achieved.
\end{lemma}
\begin{proof}
    Let $\amat \in [0,1]^{n \times m}$ be the matrix-representation of the attacker's $W$-measurable prediction rule $\varphi_R$. 
    With probability one, we have that $\accr(\varphi_R) = \frac{1}{n} \tr(\pmat\amat^\top)$. From this, it follows that
    \[
        \E\bigl[\accr(\varphi_R) \,|\, W\bigr]
        = \E\left[\frac{1}{n} \tr(\pmat\amat^\top) \,\bigg|\, W\right]
        = \frac{1}{n} \tr(\E[\pmat \,|\, W] \amat^\top),
    \]
    where the final equality follows from the fact that $\amat$ is $W$-measurable.
    Next, since $\E[\pmat \,|\, W]$ is a row-stochastic matrix, the same argument as in \Cref{randomUserAccuracyBound}, it follows that $\frac{1}{n} \tr(\E[\pmat \,|\, W]\amat^\top) \leq \| \E[\pmat \,|\, W]\|_{\infty, 1}$, as required. To find the prediction rule $\varphi^*_R$ we use the same argument as in \Cref{randomUserAccuracyBound} and define a matrix $\amat$ such that $\amat[i,o] > 0$ implies $i \in \arg\!\max_j\E[\pmat \,|\, W]$. 
\end{proof}

\section{Matching model accuracy bounds}
In this section, we provide an information theoretic upper bound on the accuracy that any attacker can achieve in the full-information setting. 
Recall that in the random-user setting, we sample a representation matrix $\pmat$ from the distribution $\pmatdist$, the attacker formulates a (possibly randomized) prediction rule $\varphi_M : \ospace^n \to \ispace^n$ based on their knowledge of $\pmat$. 
Then we let $\pi : [n] \to [n]$ be a permutation of $[n]$ chosen uniformly at random, $O_{1:n} \in O^n$ be a vector of independent observations with $O_i \sim \pmat[\pi(i),:]$, and define the accuracy random variable by $\accm(\varphi_M) = \E\left[\frac{1}{n}\sum_{i=1}^n \ind\{\varphi_M^i(O_{1:n}) = \pi(i) \,\bigg|\, \pmat \right]$.
Our accuracy bound for the matching setting is given below:
\begin{lemma}\label{lemma:reidrisk-matching}
    Let $\pmat \sim \pmatdist$ be a random representation matrix and let $\varphi_M : \ospace^n \to \ispace^n$ be the matching rule constructed by an attacker. Then with probability one over $\pmat \sim \pmatdist$, we have
    \[
    \accm(\varphi_M) \leq \frac{m}{n} - \frac{1}{n}\sum_{o \in \ospace} \prod_{i \in \ispace} (1 - \pmat[i,o]).
    \]
\end{lemma}
\begin{proof}
    Let $\pmat$ be sampled from $\pmatdist$ and $\varphi_M : \ospace^n \to \ispace^n$ be the matching rule constructed by the attacker based on their knowledge of $\pmat$.
    Next, let $O^*_{1:n}$ be independent representations with $O^*_i \sim \pmat[i,:]$, let $\pi : [n] \to [n]$ be a permutation of $[n]$ chosen uniformly at random, and define $O_{1:n}$ by $O_i = O^*_{\pi(i)}$.
    Next, for each observation $o \in \ospace$, let $S_o = \{i \,:\, O_i = o\}$ denote the set of indices $i$ for which $O_i = o$.
    Since $\pi$ is uniformly random and $O_1, \dots, O_n$ are independent with $O_i \sim \pmat[\pi(i), :]$, we have that
    \begin{align*}
    \accm(\varphi_M)
    &= \E\left[
        \frac{1}{n} \sum_{i=1}^n \ind \{\varphi_M^i(O_{1:n}) = \pi(i)\}
        \,\bigg|\,
        \pmat
    \right] \\
    &= \E\left[
        \frac{1}{n} \sum_{o \in \ospace} \sum_{i \in S_o} \Pr(\varphi_M^i(O_{1:n}) = \pi(i) \,|\, O_{1:n}, O^*_{1:n}, \pmat)
        \,\bigg|\,
        \pmat
    \right],
    \end{align*}
    where the second equality follows from breaking the sum over $i$ into a sum over $o \in \ospace$ and $i \in S_o$, and adding an inner expectation conditioned on $O_{1:n}$, $O^*_{1:n}$, and $\pmat$.
    The key idea is that, conditioned on $O_{1:n}$ and $O^*_{1:n}$, the permutation $\pi$ is still random, but $\varphi_M^i(O_{1:n})$ is fixed, which implies that $\varphi_M^i(O_{1:n})$ cannot be correct with too large of a probability.
    For any $i \in S_o$, we have that 
    \[
    \Pr(\pi(i) = j \,|\, O_{1:n}, O^*_{1:n}, \pmat) = \frac{\ind\{O^*_j = o\}}{|S_o|}.
    \]
    With this, we have
    \begin{align*}
    \accm(\varphi_M)
    &= \E\left[
        \frac{1}{n}
        \sum_{o \in \ospace : S_o \neq \emptyset}\,
        \sum_{i \in S_o}
        \frac{\ind\{O^*_{\varphi_M^i(O_{1:n})} = o\}}{|S_o|}
        \,\bigg|\, \pmat
    \right] \\
    &\leq \E\left[
        \frac{1}{n}
        \sum_{o \in \ospace : |S_o| \neq 0} 1
        \,\bigg|\, \pmat        
    \right] 
    = \E\left[\frac{1}{n} \sum_o \ind\{S_o \neq \emptyset \} \,\bigg|\, \pmat \right].
    \end{align*}
    This final expression is the number of unique representations that the attacker observed, divided by the number of users $n$. Intuitively, this bound follows from the fact that for all the users that generated the same observation, the expected number of correct guesses of the attacker is at most one.
    
    To finish the proof, we compute the expected number of distinct representations that the attacker will observe. 
    We have
    \begin{align*}
        \E\left[
        \sum_{o} \ind\{S_o \neq \emptyset\}
        \,\bigg|\, \pmat
        \right]
        = \sum_o \Pr(S_o \neq \emptyset \,|\, \pmat) 
        = \sum_o \left(1 - \prod_{i=1}^n (1 - \pmat[i,o])\right).
    \end{align*}
    which implies the statement of the Lemma.
\end{proof}

\subsection{Connections between the random-user and matching
models}\label{sec:modelConnections}
In this section we study connections between the random-user and matching settings. 
In particular, we show that the matching setting is no harder for the attacker than the random-user setting: we prove that any attacker in the random-user setting can be modified to achieve the same accuracy in the matching setting.
Next, we show that there exist representation probability matrices $\pmat$ such that an optimal attacker in the matching setting can do strictly better than the optimal attacker in the random-user setting.

\begin{lemma}
    Let $\pmat \sim \pmatdist$ and $\varphi_R : \ospace \to \ispace$ be any (possibly randomized) attacker prediction rule for the random-user setting. Define $\varphi_M : \ospace^n \to \ispace^n$ by $\varphi_M(O_{1:n}) = (\varphi_R(O_1), \dots, \varphi_R(O_n))$. Then with probability one over $\pmat$, we have $\accm(\varphi_M) = \accr(\varphi_R)$.
\end{lemma}
\begin{proof}
    Intuitively, the matching rule $\varphi_M$ applies the random-user rule independently for each representation vector in $O_{1:n}$, and the expected fraction of entries it will predict correctly is equal to the expected accuracy of $\varphi_R$ in the random-user setting. Formally, let $\pi : [n] \to [n]$ be the random permutation used in the matching setting. Then we have
    \begin{flalign*}
        \accm(\varphi_M)
        &= \E\left[ \frac{1}{n} \sum_{i=1}^n \ind\left\{\varphi^i_M(O_{1:n}) = \pi(i) \right\} \,\bigg|\, \pmat \right] \\
        &= \sum_{i=1}^n \frac{1}{n} \Pr( \varphi_M^i(O_{1:n}) = \pi(i) \,|\, \pmat) \\
        &= \sum_{i=1}^n \frac{1}{n} \Pr( \varphi_R(O_i) = \pi(i) \,|\, \pmat) \\
        &= \accr(\varphi_R),
    \end{flalign*}
    where the final equality follows from the fact that the pair $(I,O)$ with $I = \pi(i)$ and $O = O_i$ is distributed identically to $I \sim \operatorname{Uniform}(n)$ and $O \sim \pmat[I,:]$, since when $\pi$ is a random permutation, $\pi(i)$ is randomly chosen uniformly at random from $[n]$.
\end{proof}

Next, we construct a distribution $\pmatdist$ over representation matrices $\pmat$ such that an attacker in the matching setting can have a constant factor higher accuracy than the best attacker in the random-user setting.

\begin{lemma}
    For any even number of users $n$, there exists a representation space $\ospace$ of size $m = \frac{3n}{2}$ and a distribution $\pmatdist$ over representation matrices $\pmat \in \reals^{n \times m}$ such that: with probability one, every $\varphi_R : \ospace \to \ispace$ has $\accr(\varphi_R) \leq \frac{3}{4}$ in the random-user setting, and there exists a rule $\varphi_M : \ospace^n \to \ispace^n$ such that $\accm(\varphi_M) = \frac{7}{8}$ in the matching setting.
\end{lemma}
\begin{proof}
    For simplicity, we construct $\pmatdist$ as a distribution supported on a single representation matrix $\pmat$.
    First, consider the case where we have only $n=2$ users, the representation space is $\ospace = \{u_1, u_2, a\}$, and the representation probability matrix is defined by
    \[
    \pmat = \begin{blockarray}{cccc}
        & u_1 & u_2 & a \\
        \begin{block}{c(ccc)}
        \text{User 1} & 1/2 & 0 & 1/2 \\
        \text{User 2} & 0 & 1/2 & 1/2 \\
        \end{block}
    \end{blockarray}.
    \]
    We have $\frac{1}{n}\Vert \pmat \Vert_{\infty,1} = 3/4$, and it follows that no prediction rule $\varphi_R$ can achieve accuracy higher than $3/4$ in the random-user setting. However, in the matching setting, if the attacker observes at least one of $\{u_1, u_2\}$, this is sufficient for perfectly matching the users, since $u_i$ is only ever generated by user $i$. When both users generate the ambiguous representation $a$, the attacker still has a $1/2$ chance to correctly identify the users (e.g., by predicting a random permutation). The probability that both users generate representation $a$ is $1/4$, and the probability that at least one of $u_1$ or $u_2$ is generated is $3/4$. It follows an attacker in the matching setting can achieve:
    \begin{align*}
    \accm(\varphi_M) 
    &= 1 \cdot \Pr(\text{$u_1$ or $u_2$ observed}) + \frac{1}{2} \cdot \Pr(\text{only $a$ observed})\\
    &= 1\cdot\frac{3}{4} + \frac{1}{2} \cdot \frac{1}{4}
    = \frac{7}{8}.
    \end{align*}

    To extend this example to any even number of users, we create $n/2$ copies of the 2-user problem as follows: Let $\ospace = \{u_1, \dots, u_n\} \cup \{a_1, \dots, a_{n/2}\}$ and define the representation probability matrix by
    \[
    \pmat[i,o] = \begin{cases}
        1/2 & \text{if $o = u_i$ or $o = a_{\lceil i/2 \rceil}$.} \\
        0 & \text{otherwise}.
    \end{cases}
    \]
    Then $\pmat$ has $\frac{3n}{2}$ columns and the maximum value in each column is $1/2$. It follows that $\frac{1}{n} \Vert\pmat\Vert_{\infty,1} = \frac{1}{n}\cdot \frac{3n}{2}\cdot \frac{1}{2} = \frac{3}{4}$. On the other hand, for any even index $i$, we know that there are exactly two entries in $O_{1:n}$ in the set $\{u_{i-1}, u_i, a_{i/2}\}$, and that these entries must correspond to users $i-1$ and $i$ (but we do not know the order). When the attacker attempts to identify users $i-1$ and $i$, they are faced exactly with the two-user problem described above, and their expected accuracy for users $i-1$ and $i$ is $7/8$. Averaging over the $n/2$ pairs of users, their overall accuracy is also $7/8$.
\end{proof}

\section{Relation to other privacy notions}
\label{sec:comparison}
In this section we give a detailed discussion of the re-identification risk introduced in \Cref{sec:model} in relation to two prior notions of algorithmic privacy: local differential privacy and $k$-anonymity. 
We show that (for appropriate parameters) both of these privacy notions are 
\emph{sufficient} to imply low re-identification risk, but neither condition is \emph{necessary} to obtain low re-identification risk in our framework. (We refer however to the discussion in Section~\ref{sec:limits} on why they may still be needed for other privacy risks.)

This shows that the re-identification risk outlined in \Cref{sec:model} is not entirely captured by either of these concepts. We conclude this section by discussing the connection between our work and the field of quantitative information flow (QIF) \cite{qif}.

\subsection{Local differential privacy}
Local differential privacy (LDP)~\cite{cormode2018privacy} is a strong privacy notion applicable to publishing user representations constructed from private information. Intuitively, it should be hard to derive the identity of a user from the output of a differentially private mechanism. In this section we prove this implication, while, at the same time, showing that local differential privacy is \emph{not necessary} for low re-identification risk. This result highlights the ability of our framework to characterize directly and sharply re-identification risks. %

\begin{definition} [Local differential privacy]
Let $\mathcal{X}$ be an arbitrary space encoding user information. A randomized algorithm $\mathcal{A} : \mathcal{X} \to \ospace$ for mapping user data to a representation satisfies $(\epsilon, \delta)$-LDP if the following holds: for all $x, x' \in \mathcal{X}$ and any set of representations $E \subset \ospace$, we have that
\[
\Pr(\mathcal{A}(x) \in E) \leq e^\epsilon \cdot \Pr(\mathcal{A}(x') \in E) + \delta.
\]
\end{definition}

LDP representations as described above can be modeled under our framework as follows:
Let $\pmatdist$ be a distribution over representation matrices that samples $\pmat$ in two steps: 
First, the $n$ users generate their data $x_1$, \dots, $x_n \in \mathcal{X}$.
Second, we define $\pmat \in [0,1]^{n \times m}$ to have entries $\pmat[i,o] = \Pr(\mathcal{A}(x_i) = o)$, where the probability is only over the randomness of the mechanism $\mathcal{A}$.
Sampling $\pmat$ from $\pmatdist$ corresponds to the process generating the users data, while the matrix $\pmat$ encodes the mechanism $\mathcal{A}$'s output distribution for each user.
Since $\mathcal{A}$ is $(\epsilon,\delta)$-LDP, with probability one over the draw of $\pmat$, we have that for any users $i$ and $j$, and any representation subset $E \subset \ospace$, the following holds: 
\begin{equation}\label{eq:dpdist}
\sum_{o \in E} \pmat[i,o] \leq \delta + e^\epsilon \cdot \sum_{o \in E} \pmat[j,o].
\end{equation}
More generally, we say that any distribution $\pmatdist$ over representation matrices $\pmat$ that satisfy \eqref{eq:dpdist} with probability one is $(\epsilon,\delta)$-LDP.

The following result shows that $(\epsilon,\delta)$-LDP implies low re-identification accuracy in the random-user setting.
\begin{lemma}{\label{lemma:dpsuff}}
    Let $\pmatdist$ be any distribution that satisfies $(\epsilon, \delta)$-LDP, let $\pmat \sim \pmatdist$, and $\varphi_R$ be an attacker's prediction rule in the random user setting. 
    Then
    \[
    \accr(\varphi_R) \leq \frac{e^\epsilon + \min(n,m)\cdot \delta}{n}.
    \]
\end{lemma}
\begin{proof}
    Let $\pmat \sim \pmatdist$ be a sampled representation matrix and partition $\ospace$ into sets $\ospace_1, \dots, \ospace_n$, where $\ospace_i$ contains all representations that user $i$ generates with higher probability than any other user with ties broken in favor of the user with lower index. 
    That is,
    \begin{align*}
    \ospace_i = \{ o \in \ospace \,:\, \text{for all } j \neq i, \pmat[i,o] \geq \pmat[j,o] \text{and if } \pmat[i,o] = \pmat[j,o] \text{ then } i < j \}.
    \end{align*}
    Then we have that
    \begin{equation}\label{eq:inf1normByUser}
    \|\pmat\|_{\infty,1}
    = \sum_{o \in \ospace} \max_{i \in \ispace} \pmat[i,o]
    = \sum_{i \in \ispace} \sum_{o \in \ospace_i} \pmat[i,o].
    \end{equation}

    Now suppose that $\pmatdist$ is $(\epsilon,\delta)$-LDP. 
    Then we have that:
    \[
    \|\pmat\|_{\infty,1} = 
    \sum_{o \in \ospace} \max_{i \in \ispace} \pmat[i,o]
    \leq \sum_{o \in \ospace} (e^\epsilon \pmat[1,o] + \delta)
    = e^\epsilon +  m \delta.
    \]
    At the same time, from \eqref{eq:inf1normByUser} we have that
    \[
    \|\pmat\|_{\infty,1} 
    = \sum_{i \in \ispace} \sum_{o \in \ospace_i} \pmat[i,o]
    \leq \sum_{i \in \ispace} \left(\delta + e^\epsilon \sum_{o \in \ospace_i} \pmat[1,o]\right)
    = e^\epsilon + n\delta.
    \]
    The above arguments show that $\|\pmat\|_{\infty,1} \leq e^\epsilon + \min(n,m)\cdot \delta$. 
    From \Cref{randomUserAccuracyBound}, it follows that $\accr(\varphi_R) \leq \frac{e^\epsilon + \min(n,m)\delta}{n}$.
\end{proof}

\subsection{k-anonymity}
A process that releases anonymized data about a collection of users is said to be $k$-anonymous if the information released for each user cannot be distinguished from at least $k-1$ other users who also appear in the release. 
We say that a distribution $\pmatdist$ over representation matrices is $k$-anonymous if, with probability one over $\pmat \sim \pmatdist$, every row of $\pmat$ is a one-hot vector and appears at least $k$ times.
That is, each user $i$ is assigned a representation $o_i \in \ospace$ that is shared with at least $k-1$ other users, and their row of $\pmat$ is given by $\pmat[i, o] = \ind\{o = o_i\}$.

The following result shows that $k$-anonymity is sufficient to limit an attacker's accuracy to $1/k$ in the random-user setting.

\begin{lemma}
Let $\pmatdist$ be any distribution that satisfies $k$-anonymity, let $\pmat \sim \pmatdist$, and $\varphi_R$ be an attacker's prediction rule in the random user setting. Then
\[
\accr(\varphi_R) \leq \frac{1}{k}.
\]
\end{lemma}
\begin{proof}

    Let $\pmat \sim \pmatdist$ and $o_1, \dots, o_n \in \ospace$ be the corresponding representations assigned to each user.
    For each representation $o~\in~\ospace$, let $\ispace_o = \{i \in \ispace \mid o_i = o\}$ denote the set of users assigned that representation.
    From the $k$-anonymity condition, we are guaranteed that either $|\ispace_o| = 0$ or $|\ispace_o| \geq k$.
    It follows that there are at most $n/k$ observations $o$ for which $|\ispace_o| > 0$.
    This implies that
    \[
    \|\pmat\|_{\infty, 1}
    = \sum_{o \in \ospace} \max_{i \in \ispace} \pmat[i,o]
    = \sum_{o \in \ospace} \ind\{|\ispace_o| > 0\}
    \leq \frac{n}{k}.
    \]
    By \Cref{randomUserAccuracyBound}, it follows that $\acc(\varphi_R) \leq \frac{1}{k}$, as required.
   
\end{proof}

\subsection{LDP and k-anonymity are not necessary conditions for low re-identification risk}

In the previous two subsections we showed that, for appropriate parameter settings, $(\epsilon,\delta)$-LDP and $k$-anonymity both imply that an attacker in the random-user setting has low accuracy. In this section, we show that neither condition is necessary.

\begin{lemma}\label{lem:dpKanonNotNecessary}
    There exist distributions $\pmatdist$ such that with probability one, every attacker has $\accr(\varphi_R) \leq \frac{2}{n}$ and $\pmatdist$ is not $(\epsilon,\delta)$-LDP unless $\delta = 1$ and not $k$-anonymous for any $k$.
\end{lemma}
\begin{proof}
    Let $\ospace = \{1,2\}$ and let $\pmat$ have entries given as follows: for each user $i \in [n]$, define
    \[
    \pmat[i,1] = 1 - \frac{(i-1)}{n-1} 
    \qquad\text{and}\qquad 
    \pmat[i,2] = \frac{(i-1)}{n-1}.
    \]
    Not all of the rows of $\pmat$ are one-hot, so it does not satisfy the $k$-anonymity requirement. 
    Next, we have that $\pmat[1,2] = 0$ while $\pmat[n,2] = 1$, which implies that $\pmat$ only satisfies the $(\epsilon,\delta)$-LDP constraint when $\delta = 1$. 
    Finally, we have that $\|\pmat\|_{\infty,1} = 2$ and by \Cref{randomUserAccuracyBound} it follows that $\accr(\varphi_R) \leq \frac{2}{n}$.
\end{proof}

\subsection{Mutual Information}
Another view of re-identification risk can be obtained from the field of information theory. Given a joint pair of random variables $(X, Y)$,  we are interested in measuring how much \emph{information} does $Y$ encode about $X$. For our random user model, this can be translated to measuring how much information the representation $O$ provides about the identity random variable $I$. This concept is formalized by the conditional mutual information \cite{elementsofit} $\text{MI}(I; O| \pmat)$  for the full information scenario and by $\text{MI}(I; O | W)$ for the partial-information scenario. This metric was in fact used in prior work~\cite{topics-explainer} to quantify the re-identification risk of the Topics API.
One can use the celebrated Fano's inequality \cite{elementsofit} to show that.
\begin{lemma}
Under the random-user model we have
\begin{align}
    \accr(\varphi_R)\le \frac{1+\text{MI}(I;O|\pmat)}{\log(n)}\,,
\end{align}

\end{lemma}
It is worth noticing that the dependency on the number of users $n$ here is logarithmic as opposed to that of Lemma~\ref{randomUserAccuracyBound} where the dependency is linear. This is an exponential improvement and demonstrates that our framework can better capture re-identification risks.

\subsection{Quantitative information flow}
Quantitative information flow \cite{qifsmith, qif} (QIF) is a different framework for analyzing the privacy vulnerability of a system. QIF is specified by a space of secrets $\mathcal{S}$, an output space $\ospace$, a (possibly randomized) channel $C \colon \mathcal{S} \to \ospace$ assumed to be known to an adversary, and a gain function $g \colon \mathcal{W} \times \mathcal{S} \to \mathbb{R}$, where $\mathcal{W}$ is an adversary's space of strategies that may coincide with $\mathcal{S}$ depending on the scenario.

QIF assumes there is a secret $s$ sampled from a known distribution $\pi$ and that the adversary observes $o = C(s)$. The goal of the adversary is, given $o$, to learn about $s$. The gain function may then capture the reward of an adversary predicting secret $s'$, here a reward function $r(s', s)$. Note that an adversary with access to the channel, given an output $o$, can always predict the secret $s'$ that maximizes their posterior reward 
$$R(o) = \max_{s'}\sum_{s \in \mathcal{S}} r(s', s)P(s | o).$$

The privacy vulnerability of a channel may be seen in QIF as the expected posterior reward $\E_\pi[R(o)]$. For our full-information setting, the known representation matrix $\pmat$ corresponds to the channel, the identity space $\ispace$ is the secret space and the reward function $r(s', s) = 1$ if $s' = s$ and it is $0$ otherwise. That is, the adversary is only rewarded if they predict the correct user in one try. It is known \cite{qif} that for this scenario the expected posterior reward corresponds to the so-called Bayes vulnerability and it satisfies:
\begin{equation*}
    \E[R(o)] = \frac{1}{n}\|\pmat\|_{\infty, 1}
\end{equation*}

That is, our full-information setting is an alternative formulation of QIF as a hypothesis testing framework. To the best of our knowledge, we are not aware of a partial information QIF formulation that fully matches the random-user or matching scenarios although we are actively exploring ways to use advanced concepts in QIF such as the internal fixed-probability choice model \cite{qif} to establish a similar connection.

\section{Case Study: The Topics API}\label{sec:case-study}
As mentioned in the introduction, we will use our framework to provide an analysis of the re-identification risk in context of the Topics API~\cite{topics-github} of the Privacy Sandbox~\cite{privacysandbox}. Here, we introduce the Topics API using the framework of section~\ref{sec:model}.

The Privacy Sandbox~\cite{privacysandbox} is a series of proposals to enable online advertising while  limiting cross-site tracking on the web. We will focus on Interest Based Advertising (IBA) use case of the Privacy Sandbox. IBA is a sector of online advertising in which ad-tech providers build models of the users' interests in an effort to show them relevant ads. For instance, people interested in a car may be served car ads even on unrelated pages. 

Historically, IBA has been enabled through third-party cookies. These serve as a cross-site user identifier, allowing ad techs to keep track of the sites a user has visited and build an interest profile based on their browsing history. This cross-site tracking is in direct contrast with the goals of the Privacy Sandbox, which has led Chrome to announce the Topics API to support IBA without relying on cross-site tracking.

The Topics API works as follows (we refer to the specifications in~\cite{topics-github}): every week the {\em browser} builds an interest profile of the user, in the form of selecting top five topics from a fixed topics taxonomy, $\topics$. Importantly, this profile is kept on the browsers and is not shared with others. 

Whenever a website wants to show an ad, the browser shares a topic selected uniformly at random from one of the top 5 topics in the profile of the previous week with the ad tech provider (additionally with some probability $p$, it may simply return a uniformly at random topic from $\topics$). Crucially, for every user, the topic sampled for a website is fixed for an entire week, and the samples on two different websites are independent.\footnote{The API actually returns, on top of the current sampled topic, a cached result for the output of the previous 2 weeks for the caller. We omit this detail from the modeling as observing $r$ consecutive weeks of Topics, simply corresponds to performing calls for $r+2$ weeks in our model.} See Figure~\ref{fig:cookievtopics} for an example of the Topics API. A detailed specification of the API can also be found in Algorithm~\ref{alg:topics}.

\begin{algorithm}
\caption{Topics API. 
\newline {\bf Input:} Topics $\topics$, probability to return a random
topic $p$. }
\label{alg:topics}
\begin{algorithmic}
\STATE {\bf On device:} Select set $S$ of top 5 most popular topics for this client.
\STATE \textbf{On call GetTopic()} from website $w$ on week $s$ and user $u$:
\STATE Seed the random number generator with $w, s, u$.
\STATE Flip coin with heads probability $p$.
\IF{Heads}
\RETURN Element of $\topics$ chosen u.a.r. 
\ELSE
\RETURN Element of $S$ chosen u.a.r. 
\ENDIF
\end{algorithmic}
\end{algorithm}

\begin{figure}
    \centering
    \begin{tabular}{c|c}
     \tiny{(a)} \includegraphics[width=.46\linewidth]{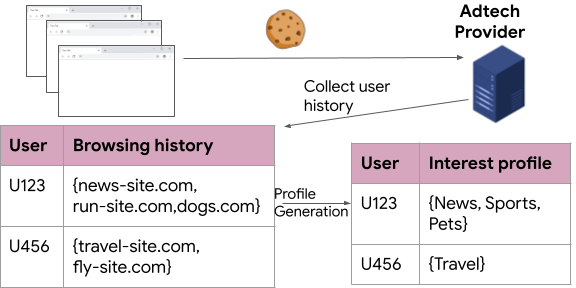} & 
    \tiny{(b)}\includegraphics[width=.46\linewidth]{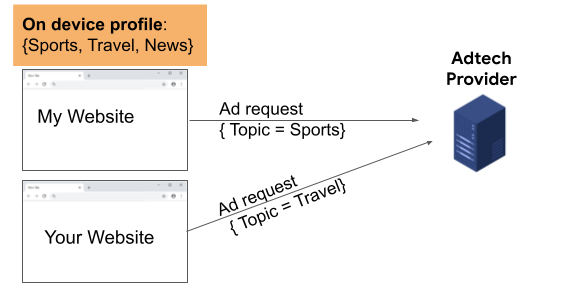}\\
    \end{tabular}
    \caption{Comparison of IBA under (a) third party cookies and (b) the Topics API. In the former, the ad tech provider gets to build the browsing history of a user. In the latter, the ad tech provider only observes a single topic for this user.}
    \label{fig:cookievtopics}
\end{figure}

\paragraph{Re-identification risks of the Topics API} 
Compared to third-party cookies, the Topics API has a significantly lower risk of re-identifiability (the former  guarantees re-identifiability by nature of being a persistent cross-site identifier). The goal of this section is to formally measure this risk in the case of misuse of the API. 

Consider a user that visits two sites every week. Over time, even with the randomization, the sequence of topics observed on website 1 will be similar to the sequence of topics observed on website 2 (for instance, we expect an exact match in $(1-p)/5$ fraction of the weeks). Thus two sites could try to collude to use the Topics API to link the identity of a user across them.%
\footnote{We restrict our analysis to two sites only for simplicity of exposition, as two sites are sufficient to elucidate the re-identification risk. We refer to  section~\ref{sec:limits} for the limitations of our work.} 

Figure~\ref{fig:topic_linking} shows how this attack may happen. We now formalize the re-identification risks of the Topics API from the perspective of website 2 as an attacker colluding with website 1. 

Let  $N = |\topics|$ and $\ospace = \topics^r$. That is if $o \in \ospace$ then $o = (o^1, \ldots, o^r)$ where $o^s$ corresponds to the topic returned by the API on week $s$. We begin by modeling a single representation of the Topics API.  Based on Algorithm~\ref{alg:topics}, for every user $i \in \ispace$, the topic returned by the API at round $s$ depends on the set of top topics $S_i^s$ associated with the user. Having fixed this set, the topic selection of the API at round $s$ can be modeled by a matrix $\pmat_s$ given by:
\begin{equation}
    \pmat_s[i,o^s] = \left\{ \begin{array}{cc}
    q_{\text{in}} := (1- p)/5 + p/N& o^s \in S_i^s \\
    q_{\text{out}} := p/N & o^s \notin S_i^s
    \end{array}\right.
\label{eq:topic_matrix}
\end{equation}
It is easy to see that the representation matrix for $o \in \ospace$ is given by
\begin{equation*}
    \pmat[i,o] = \prod_{s=1}^r \pmat[i, o^s]
\end{equation*}
\begin{figure}
    \centering
    \includegraphics[scale=0.3]{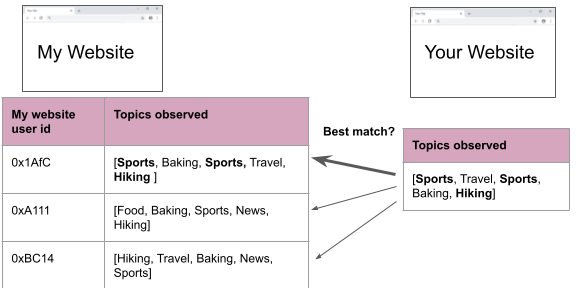}
    \caption{Example of two sites using the sequence of topics returned by the API to potentially re-identify a user across two sites. In the example the user on website 2 shares 3 topics with a user on website 1. This has the most matches across all three users.}
    \label{fig:topic_linking} 
\end{figure}

Thus far we have assumed the top set of topics $S_i^s$ to be fixed. However,  notice that these are not observable by the websites directly, since they can only see the samples from the top sets. Moreover, the sets $S_i^s$ are determined by the behaviors of the users which we encode as a probabilistic process in the distribution $\pmatdist$. 

Formally, the matrix $\pmat$ is a random variable sampled from a (latent) distribution $\pmatdist$ which encodes the way that the top set of topics $S_i^s$ is generated. Clients of the Topics API (i.e., websites) learn about the sampled matrix $\pmat$ through the observation of samples of the representations in their site.

With this modeling in mind, we see that this re-identification risk can be defined in terms of the (partial) information random-user model:
\begin{enumerate}
    \item A matrix $\pmat$ is sampled from $\pmatdist$.
    \item A sample of representations $W = (W_1, \ldots, W_n) \sim \pmat$ ($W_i$ corresponds to the representation of user $i$) is obtained from website~$1$ and shared with colluding website~2.
    \item Given the representation $O$ of a random user on website~2, the attacker must find a $W$-measurable function to predict the identity of that representation to match that user on website~1.
\end{enumerate}

\paragraph{Modeling the distribution $\pmatdist$} A natural question is: how can we model the distribution $\pmatdist$? And, can we define the optimal action $\varphi$ taken by an attacker given their knowledge of $\pmatdist$? Below we propose a natural parametric family of distributions $\pmatdistfamily$, and an efficient way for the attacker to estimate their parameters based only on the samples known. We also  describe a simple to implement optimal attack for this parametric family. In section~\ref{sec:exp-mi} we verify empirically that our assumptions on $\pmatdistfamily$ closely match the observations from current web traffic. 

Note that since $\pmat$ is fully determined by the top sets $S_i^s$, we can equivalently define $\pmatdistfamily$ as a family of joint distributions over the sequence of random variables $(S_i^s)_{i\in[n],s \in [r]}$. We denote by $\mathbb{P}$ the probability measure induced by $\pmatdist$ on $S_i^s$. A distribution $\pmatdist$ belongs to $\pmatdistfamily$ if it satisfies the following conditions:
\begin{enumerate}
    \item All users have the same distribution for variable $S^s_i$ for a given time $s$ (i.e.  the top topics of the user are sampled independent and identically distributed from the same distribution, but naturally users can have different sampled top topics.)
        
    \item For every user, samples of top sets are independent across time (but not necessarily identically distributed).
\end{enumerate}
It is not hard to see that the assumptions on $\pmatdistfamily$ imply that distribution $\pmatdist \in \pmatdistfamily$ if and only if there exist distributions $\pmatdist_1, \ldots, \pmatdist_r$ such that
\begin{equation*}
    \pmatdist(\pmat) = \prod_{s=1}^r \pmatdist_s(\pmat_s).
\end{equation*}
For each distribution $\pmatdist_s$ we will be interested in the following 
\begin{equation*}
    \pmatdist_s(\pmat_s[i,o] = q_{\text{in}}) = \mathbb{P}(o \in S_i^s) = p_s[o]
\end{equation*}
These parameters represent the probability that topic $o$ is part of the top $5$ topics of a user. The following lemma, which is proved in the Appendix, shows how we can use observations from user representations to estimate these terms.
\begin{lemma}
\label{lemma:estimation}
Let $W_1^s, \ldots, W_n^s$ be a sample of topics on website 1, let $N = |\topics|$ and $\delta > 0$. Let also
$$\widehat{p}_s[o] = \frac{1}{n} \sum_{i=1}^n \frac{\ind_{W_i^s = o} - q_{\text{out}}}{q_{\text{in}} - q_{\text{out}}},$$
where $\ind_{W_i^s = o}$ is 1 iff $W_i^s = o$. With probability $1 - \delta$ uniformly across all topics $o \in \topics$ the following inequality holds:
\begin{equation}
    |\widehat{p}_s[o] - p_s[o]| \leq \frac{1}{q_{\text{in}} - q_{\text{out}}} \sqrt{\frac{\log(2 N/\delta)}{2 n}}.
\end{equation}
\end{lemma}
Let us now proceed to identify the optimal attacker under the assumptions on $\pmatdist$. We begin by discussing a natural attacker to derive some intuition into the results of our main theorem of the section. 

\begin{example}[Hamming attack]\label{ex:hamming}
Given a representation \\ $o = (o^1, \ldots, o^r)$ and representations $W = (W_i^s)_{i \in [n], s\in [r]}$ a na\"ive attacker will assign representation $o$ to a user $i$ such that $|\{s : W_i^s \neq o^s\}$ is minimal. That is, it will naturally try to find the user that minimizes the Hamming distance between $W_i$ and $o$.
\end{example}
 While the Hamming distance attack matches the intuition that the most likely user that generates a representation is that with the largest overlap on the topic sequence of website 1, one should keep in mind that not all topics are the same. Indeed, a match on a very unpopular topic should be worth more than a match on a popular topic. It is not hard to see that the parameters $p_s[o]$ are a proxy for the popularity of a topic across the population. Thus an optimal attacker should leverage this information. The following theorem, proved in the Appendix, formalizes this intuition.
\begin{thm}[Asymmetric Weighted Hamming Distance Attack]
\label{thm:optimal_attacker}
    Given a representation $o=(o^1, \ldots, o^r)$ on website~2, and representations  $W = (W_i^s)_{i\in [n], s\in[r]}$ on website~1.  An attacker that wants to maximize its accuracy under the partial information setting selects the identity of the user that minimizes the following \emph{asymmetric weighted Hamming distance}
    \begin{flalign*}
         -\!\!\!\sum_{s : W_i^s = o^s}\!\!\! \log\left(q_{\text{out}} + \frac{(q_{\text{in}}-q_{\text{out}})q_{\text{in}} p_s[o^s]}{q_{\text{out}} + (q_{\text{in}} - q_{\text{out}})p_s[o^s]}\right)
        -\sum_{s : W_i^s \neq o^s} \!\!\!\!\!\log\left(q_{\text{out}} + (q_{\text{in}} - q_{\text{out}}) \mathbb{P}(o^s \in S_i^s| W_i^s \in S_i^s)\right).
    \end{flalign*}
\end{thm}
Notice how the attack in Theorem~\ref{thm:optimal_attacker} can be seen as an asymmetric variant of the simple Hamming distance attack described above. It is important to note that while there may likely be several heuristics for utilizing the Topics API signal for re-identification our framework has allowed us to derive  --- from first principles --- a simple optimal algorithm under some basic assumptions. We expect that future work on understanding the privacy of the Topics API can be done by  relaxing some of these assumptions.

\section{Empirical analysis}
\label{sec:exp}

In this section we give empirical evaluations of our model for re-identification risk  
on two real-world datasets: (1) A Google proprietary dataset containing de-identified user data from a simulation of the Chrome Topics API and (2) the public Million Song Dataset~\cite{MSD}.
In order to foster the reproducibility of our results, we released our code open source.\footnote{The code is available at: \url{https://github.com/google-research/google-research/tree/master/re_identification_risk}}
This section proceeds as follows. 
In sections~\ref{sec:exp-attacks}-\ref{exp-attack-algorithms}, we present empirical implementations of the re-identification attacks presented before and more advanced machine learning heuristics. 
Next, in section~\ref{sec:exp-attacks-results}, we show our results for the re-identification attack for Topics API.
In section~\ref{sec:exp-mi} we use machine learning methods to validate the hypothesis made in section~\ref{sec:case-study}. To do so, we further extend and validate the analysis of mutual information of Topics published previously~\cite{topics-explainer}.
In section~\ref{exp:msd}, we show our results for the re-identification attack for Million Song Dataset.

\subsection{Re-identification task for Topics API  in random-user model}

In this section we present our re-id attack on the Topics API. As a first step, we describe the data used

\paragraph{Chrome Topics API data}
In our empirical analysis we simulate the output observed by an adtech from the Topics API for a set of users over a period of time.  This is achieved using a Google proprietary dataset of de-identified user  browsing histories. Starting from this dataset, we run the Topics API algorithm for a sample of such users and simulate the output observed for two sites by the adtech from the Topics API sampling distribution. Our observation period consists of 8 intervals of 7 days of traffic shifted by 3 days each -- that is we consider intervals $[1,7], [4,10], [7, 13], \ldots$.
Each interval (hence-forth epoch) is used to establish the top $k$ topics of the user using the API topic model.\footnote{The real Topics API has disjoint epochs of 7 days, here we simulate training over overlapping periods because our analysis is limited for privacy reasons to 4 weeks of data. This allows us to simulate longer topics sequences. We do not observe a significant difference in the results.} We restrict our analysis to the set of users that are observed in every epoch. 

To be consistent with the current Topics API specification~\cite{topics-github}, for users with fewer than $k$ topics in an epoch, we pad the top topics with random topics. Moreover, we set $k=5$ for the number of topics and use $p=5\%$ for the probability of releasing a random topic instead of an organic one, as currently implemented in the Chrome browser.

\label{sec:exp-attacks}

\paragraph{Methods}
We now present the methodology used to establish the accuracy of empirical re-identification attacks on the Topics API. We focus on the random-user model and consider different attack algorithms presented in section~\ref{exp-attack-algorithms}. We extract a target dataset of $10$ million users chosen uniformly at random that are used in the re-identification attack analysis as the set $\ispace$ over which to re-identify the user and simulate the observation on website 1 of all their topics sequence for $r$ epochs ($r \in [1,8])$. Then we repeat $10,000$ times a uniformly random draw of a user from the set $\ispace$, generate the $r$-length sequence on website 2 for such user and verify if the attack algorithm matches it correctly to the sample in website 1. 

\subsection{Attack algorithms}
\label{exp-attack-algorithms}
We simulate the following three attacks methods: 
the {\bf Unweighted Hamming} attack and the {\bf Asymmetric Weighted Hamming} attack as well as the {\bf Neural Network} attack. We now describe each method before presenting our results in section~\ref{sec:exp-attacks-results}.

\subsubsection{Unweighted Hamming Attack}
This method is an exact implementation of the simple attack presented in Example~\ref{ex:hamming}. 

\subsubsection{Asymmetric Weighted Hamming Attack}
This method is a simplified implementation of the attack given in Theorem~\ref{thm:optimal_attacker} which is optimal under the assumptions described above. 

In our experiments, we further make the approximation of assuming that $\forall W_i^s\not=o_s,\mathbb{P}(o^s \in S_i^s| W_i^s \in S_i^s)$ is a function only depending on $o^s$. 
Then we show that $\forall W_i^s\not=o_s,\mathbb{P}(o^s \in S_i^s| W_i^s \in S_i^s) = \frac{4p_s[o^s]}{5-p_s[o^s]}$ (see Lemma~\ref{lemma:alpha_calculation} for more details).
This reduces the parameters of the model to be estimated to only $p_s[o^s]$. Moreover, given that we empirically observe $p_s[o^s]$ to be very close in every period we further assume $p_s[o^s] = p[o^s]$.

\subsubsection{Neural Network Attack} \label{exp-attacks-dnn}
In addition to the previous methods, we also implement a heuristic attack method based on a deep neural network.
Our method is general and works on an arbitrary embedding representation for the topics sequence of a user. In section~\ref{exp:s2q} we show how we obtain such embedding from sequence to sequence models while in this section we focus on how to use any embedding for matching users.  

We train a deep neural network which takes two sequence embeddings as inputs and outputs a similarity score in $[0,1]$, indicating the predicted probability that the two sequences are from the same user.
Our network structure is similar to that of the  Grale infrastructure~\cite{halcrow2020grale}. The detailed network structure is presented in Figure~\ref{fig:grale}.

\paragraph{Training process.}
We sample $20$ million random users (different from the target set $\ispace$  used in the re-identification task). 
For each sampled user $u$, we simulate a pair of topics sequences where both sequences are from $u$ and we regard it as an example from the class of correct matching.
We also create $10$ pairs of sequences where the first sequence of each pair is the sequence from $u$ and the second sequence is from a random user $v\not =u$, and we regard each pair as an example from the class of incorrect matching.
Each training example is a pair of sequences generated by the above procedure and each sequence is embedded using a sequence model.
The training objective of the neural network is to minimize the binary cross entropy loss.
\paragraph{Re-identification inference.}
Finally, to match a user sample using the neural network, given a user sequence $A$, we enumerate every sequence $B$ from the target dataset and feed $(A,B)$ to the neural network.
We choose the sequence $B^*$ as the re-identification output where the sequence $(A,B^*)$ maximizes the predicted probability given by the neural network.

\subsection{Results for the re-identification attack}
\label{sec:exp-attacks-results}

We report the estimation of the probability of each attack algorithm to correctly re-identify the random user over the $10$ million target set of users. In our analysis we study $r=1,2,4,6,8$ epochs. The main result is shown in Figure~\ref{fig:reid_10m}.

\begin{figure}
\centering
\includegraphics[width=0.45\textwidth]{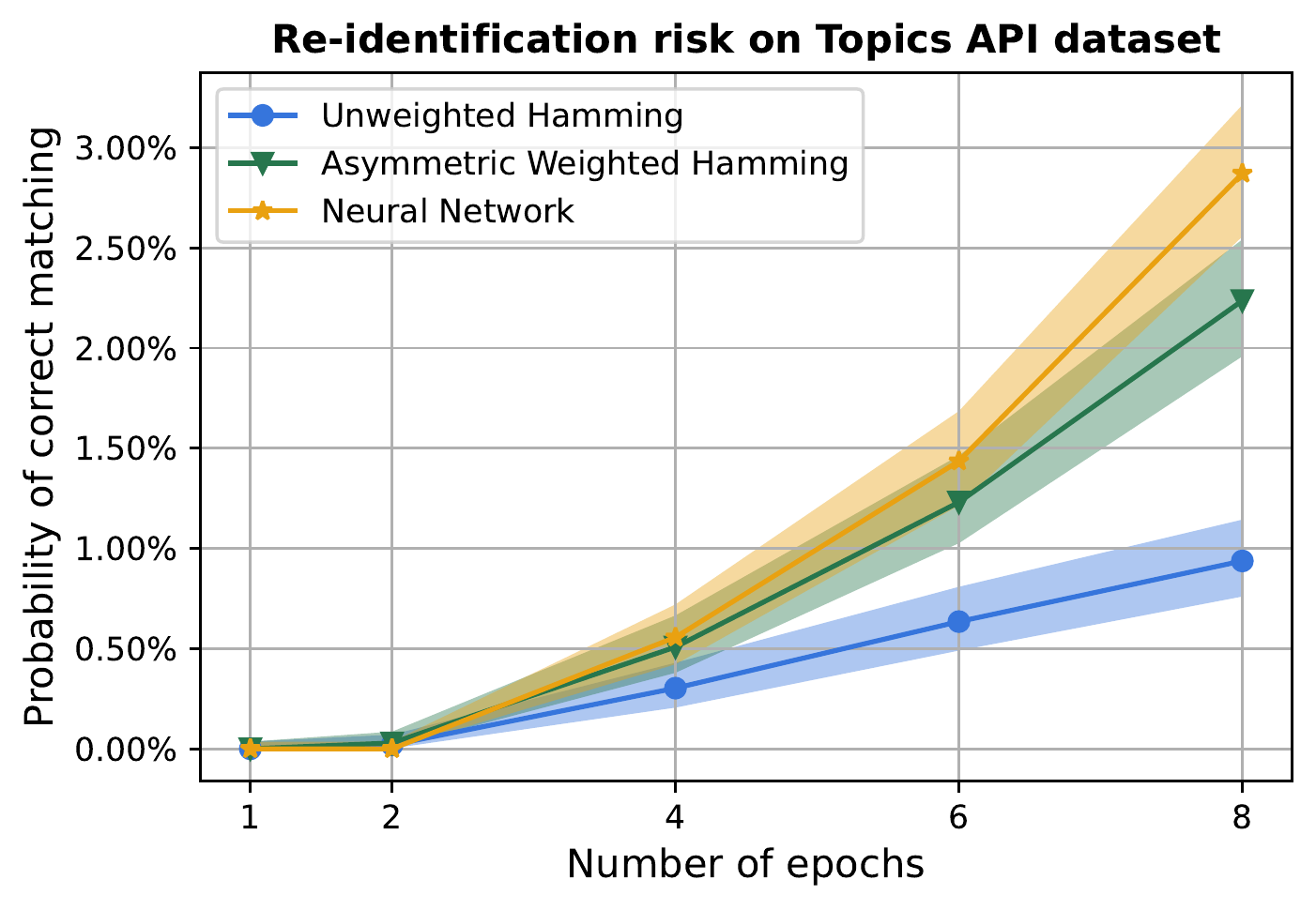}
\caption{Probability of a correct cross-site match depending on the number of epochs observed. The $95\%$ confidence intervals are reported. Notice how, even after 8 epochs, the probability of correct re-identification is below 3\%. \label{fig:reid_10m}}
\end{figure}

As expected, we observe that as the number of epochs increases there is an increased probability of correctly matching the user across the two sites (as more information is available to the attacker). Notice also how the more sophisticated Neural Network attack outperforms all methods but has performance close to our implementation of the provably good attack given by Theorem~\ref{thm:optimal_attacker}. This is a further confirmation of the validity of the simplifying assumptions made in our theoretical study. On the other hand, the simple unweighted Hamming attack performs less well. We also observe the gap between algorithms increases with the number of epochs. This is expected as more epochs of observation allow the advanced neural network algorithm to learn more correlations across the data.
Overall, we observe that even after $8$ epochs, the probability of correct re-identification is below $3\%$.  

\newpage
\subsection{Validation of the assumptions: Analysis of Mutual Information}
\label{sec:exp-mi}

\begin{wraptable}{R}{0.63\textwidth}
\centering
\small
\begin{tabular}{lll}
\text{Hyperparameter} & Transformer & Transformer-LSTM \\ \hline
Encoder & Transformer & Transformer \\
Decoder & Transformer & LSTM \\
Attention dropout rate & 0.1 & 0.1 \\
Attention layer size & 1,024 & 1,024 \\
Dropout rate & 0.1 & 0.1 \\
Embedding size & 1,024 & 1,024 \\	
MLP dimension & 4,096 & 4,096 \\	
Number of attention heads & 8 & 8 \\
Number of encoder layers & 6 & 6 \\
Number of decoder layers & 6 & 8 \\
Decoder Hidden dimension & 1,024 & 1,024 \\
Training batch size & 1,536 & 1,536 \\
Total number of parameters & - & 300M \\
\hline
\end{tabular}
\caption{\label{tab:hyperparameters} Hyper-parameters of the S2S Model Architectures.}
\end{wraptable}

\begin{figure}[b]
\centering
\includegraphics[width=0.9\textwidth] {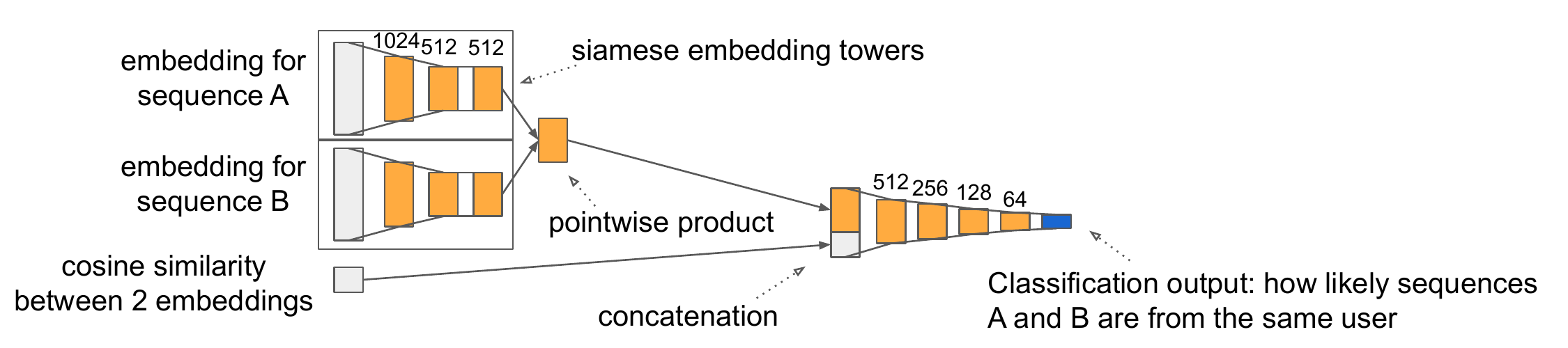}
\caption{The structure of the network for re-identification attack. %
The embeddings of each sequence are computed by feeding the sequence to the trained S2S model (Transformer-Only) described in Section~\ref{sec:exp-mi} and concatenating the hidden representations of the final encoder layer at each week. %
The number in the figure indicates the size of each layer. If not specified, each layer is a Fully Connected layer + ReLU activation. \label{fig:grale} }
\end{figure}

In this section, our aim is to measure the validity of a key assumption made in Theorem~\ref{thm:optimal_attacker} which we simulated: that of the cross-time independence of topics.

To do this we  measure  mutual information $M^*= I(A^1,A^2, \cdots A^r;$ $B^1,B^2, \cdots B^{r})$, for $A=A^{1},A^{2},\cdots A^{r}$ and $B=B^{1},B^{2},\cdots,B^{r}$ being the two sequences of topics observed from a given user on two websites.
Notice that, as shown in~\cite{topics-explainer} in case of cross-time independence this $M^*$ reduces to the easy to compute $\sum^r_{s=1} I(A^s; B^s)$. By estimating $M^*$ and showing that it is close to $\sum^r_{s=1} I(A^s; B^s)$ we verify the accuracy of the assumption.

In general estimation of $M^*$ is a non-trivial task as the distribution of long sequences of topics has an exponentially growing support. In this section we use advanced ML techniques to tackle this challenge allowing us to extend the previously published work in~\cite{topics-explainer} beyond the analysis of $2$ epochs of data.

The rest of the section presents the ML model and how we used it to estimate the mutual information.

\subsubsection{Sequence to Sequence Model}
\label{exp:s2q}

We use state-of-the-art sequence to sequence models (S2S)~\cite{sutskever-seq2seq} that are currently being used in a variety of machine learning applications ranging from natural language processing~\cite{transformer} to computer vision. As a byproduct of this model we develop an embedding for a topics sequence that we used in the previous section for the Neural Network attack.

First, we introduce formally the concept of sequence-to-sequence model.
A sequence-to-sequence (S2S) model~\cite{sutskever-seq2seq} assigns a probability to a sequence of target symbols $B^1, B^2, \cdots B^Y$ given a sequence of source symbols $A^1, A^2, \cdots A^X$: $P(B^1=b^1,B^2=b^2,\cdots,B^Y=b^y|A^1=a^1,A^2=a^2,\cdots,A^X=a^x)$. The general architecture of an S2S model consists of an encoder network which generates an embedding of the source sequence and a decoder network that generates the target sequence conditional on the source sequence. A further refinement is an attention mechanism~\cite{bahdanau2015} that allows the decoder to attend to specific tokens of the source sequence when assigning probabilities to each token in the target sequence. A number of architectures have been proposed for the encoder and decoder networks. In this work, we will compare two popular architectures: Transformer~\cite{transformer} and a variant consisting of a Transformer encoder and a Long Short Term Memory (LSTM) decoder~\cite{hybridmodel}.

\subsubsection{Estimation method}
Now, we present how to use S2S models to estimate mutual information.

The mutual information between sequences $A$ and $B$ can be written as a difference of two entropies~\cite{elementsofit}:
\begin{align}
    & I(A^1,A^2, \cdots, A^r ; B^1,B^2, \cdots, B^r) = \\
    \nonumber & H(B^1,B^2, \cdots, B^r) - H(B^1,B^2, \cdots, B^r|A^1,A^2,\cdots, A^r)
    \label{eq:seqmutualinformation}
\end{align}

Thus if we train two S2S models to estimate $H(B^1,B^2, \cdots B^r)$ and 
$H(A^1,A^2, \cdots A^r|B^1,B^2,\cdots ,B^r)$ respectively, we can compute an estimate of  $I(A^1,A^2, \cdots A^r ; B^1,B^2, \cdots B^r)$.

Given a training data set of tuples consisting of a user, and a pair of topics sequences associated with the user on two websites, $A=A^1,A^2,\cdots, A^r$, $B=B^1,B^2,\cdots,B^r$, we estimate a sequence-to-sequence model to predict $B$ given $A$. Over an unseen test data set consisting of $n$ tuples $\{ {A}_{i}^{1},{A}_{i}^{2},$ $\cdots {A}_{i}^r,{B}_{i}^{1},{B}_{i}^{2},$ $\cdots {B}_{i}^{r} \}_{i=1}^{n}$, we can estimate the conditional entropy of $B$ given $A$:
\begin{align}
H(B^1,B^2,\cdots,B^r | A^1,A^2,\cdots,A^r) = \\
-\frac{1}{n} \sum_{i=1}^{r} \text{log} P({B}_{i}^{1},{B}_{i}^2,\cdots {B}_{i}^{R} | {A}_{i}^{1},{A}_{i}^2,\cdots {A}_{i}^{r})
\label{eqn:conditionalentropy}
\end{align}

To estimate the unconditional entropy of the target sequence, we replace the source sequence within each tuple in the data set (both training and test) with a single special token $\mathcal{\$}$, (i.e. $A^1, A^2, \cdots, A^r$ is replaced with $\mathcal{\$}$). We can then use the above approach to train a model to estimate the unconditional entropy of the target sequence:
\begin{equation}
H(B^1,B^2,\cdots,B^r) = -\frac{1}{n} \sum_{i=1}^{r} \text{log} P({B}_{i}^{1},{B}_{i}^2,\cdots {B}_{i}^{r} | \mathcal{\$})
\label{eq:unconditionalentropy}
\end{equation}

\subsubsection{Model Details}
In our work we experiment with two separate S2S model architectures: either the vanilla Transformer~\cite{transformer}
or a variant~\cite{hybridmodel}
consisting of a Transformer encoder and an LSTM decoder. (Exact details on the architecture may be found in the above references). For estimating mutual information using a given sequence-to-sequence model architecture, we use the same set of hyper-parameters for both the unconditional and the conditional model (Table~\ref{tab:hyperparameters}).

\subsubsection{Results}

We now present the main set of results from the sequence-to-sequence model using the Transformer architecture.
The hybrid architecture achieves a very close but slightly lower mutual information than the Transformer-only model so we omit its results and use the Transformer only embeddings in our re-identification analysis.

We measured the unconditional entropy of a sequence of topics for $r$ epochs and as well as the conditional entropy (of the second sequence of the user) and their difference which is the mutual information. We report the results as bits/epoch. Using the Transformer model trained on $r=8$ epochs, we observe $6.54$, $5.45$ and $1.09$ bits/epoch of unconditional entropy, conditional entropy and mutual information, respectively. On a single epoch of data we observe $7.59$, $6.66$, $0.93$ bits for unconditional entropy, conditional entropy and mutual information, respectively.

Notice how, on average, we observe about $1.1$ bits of mutual information per epoch of observation vs $0.93$ bits of a single epoch. 
This suggests that while there is indeed some information gained by looking at sequences of topics across time, previous observations of the Topics API do not provide significant information about the topic returned in the current epoch. This validates our hypothesis that topics are close to independent across time and the model for the distribution $\pmatdist$.

\subsection{Re-identification task for Million Song Dataset and Results}\label{exp:msd}
We now present our empirical study of re-identification attacks on the publicly available Million Song Dataset~\cite{MSD}.
Similar to the study of Topics API, we focus on the random-user model. 

In this dataset, a user is represented by all songs liked by them. The dataset contains $48$ million entries on the listening activity of about a million users. The number of distinct songs is about a million.

We simulate a system that outputs a sample of $r$ songs for a user, independently, to generate two different databases.
Then, we measure the risk of re-identifying the user across the two datasets, depending on the number $r$ of independent samples observed. The results are reported in Figure~\ref{fig:reid_song}.

\begin{wrapfigure}{r}{0.5\textwidth}
\centering
\includegraphics[width=0.45\textwidth]{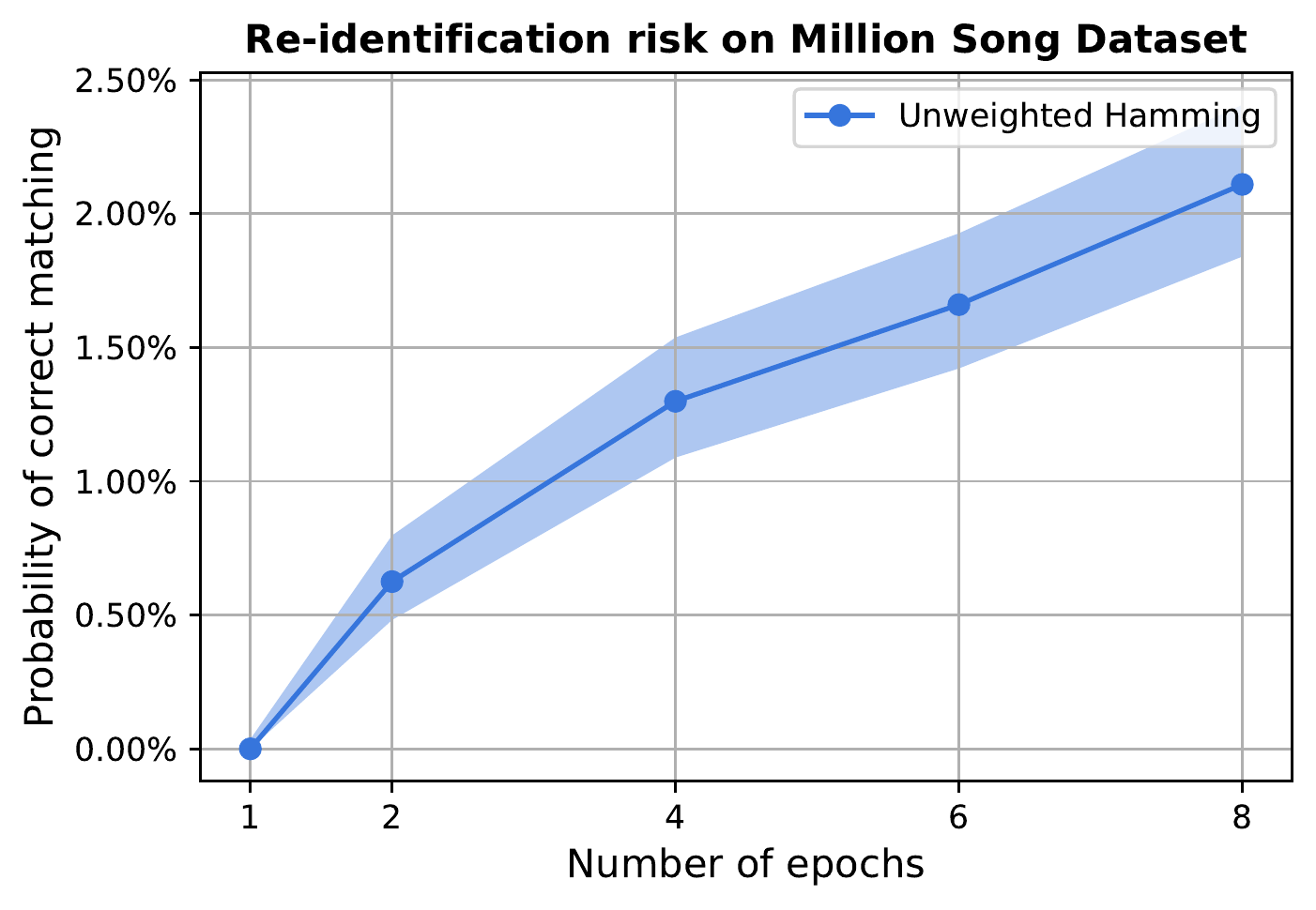}
\caption{Probability of a correct cross-site match depending on the number of epochs observed. The $95\%$ confidence intervals are reported.  \label{fig:reid_song}}
\end{wrapfigure}

Here we use the Unweighted Hamming Attack for guessing the match of the user (notice that the other two attacks were specifically designed for Topics, so they are not meaningful here). We note that observing $4$ independent random songs results in $1\%$ re-identification risk. This is an example of how our framework allows us to assess the risk in such data releases.

\section{Discussion and limitations}
\label{sec:limits}

In this paper we have presented a framework for quantifying re-identification risk. Our theoretical formulation is general enough to frame many common notions of privacy like $k$-anonymity and differential privacy, and to capture real-world examples such as the Topics API. 
Our experiments show an empirical estimation of re-id risk even on very large representation spaces. We conclude this paper highlighting some of the limitations of our work and future directions of research in this area.%

\paragraph{Scope of the privacy risk measurement}
First, it is important to observe that in our modeling we have focused exclusively, and purposefully, on re-identification risk. Preventing re-identification is only one of the many privacy and safety requirements of real-world applications. %
As real systems often require to bound additional risks (such as membership inference attacks~\cite{shokri2017membership}), we believe that privacy protections from differential privacy (DP) will be needed to complement re-identification risk analyses in many real systems. 

\paragraph{Theoretical limitations}
From a theoretical point of view, we make mild modeling assumptions on the generation process for the representations. Specifically, we assume that each user sample from $\pmat$ is drawn independently and we restrict our modeling to discrete representations. We remark that the partial information scenario we consider, and the distribution family $\pmatdistfamily$ we model, are just one of the many ways in which we could model an attacker with limited information on $\pmat$. We omit for instance, modeling the presence of more than 2 colluding clients. 

More crucially, for the upper bound results, we make a closed-world assumption:  we measure the re-identification risk under the assumption that this is the only information about a user available to an attacker. In other words, all of the side information is built into the representation matrix $\pmat$. While limiting in some situations, this assumption is reasonable in other cases (as we detail below). Moreover, our lower bounds hold even without this assumption (i.e., additional side information can only increase the risk). 

\paragraph{Limitations of the Topics API study}
In generating the data for the Topics API study, we assumed that users visit both websites in all of the epochs; thus we specifically avoided dealing with gaps in the data. This is unlikely to be the case in a real world deployment of such an API. 
An additional concern is that we used overlapping time periods to simulate longer time horizons, to deal with limitations of the data we had on hand. 
On the modeling side, while we have tried to develop powerful re-identification methods, future researchers may be able to find more powerful approaches, increasing the empirical re-identification risk.  

Finally, in our analysis we have focused exclusively on the Topics API output in isolation, ignoring other sources of side information  that an attacker might have. Considering additional sources of information is beyond the scope of our work and would require incorporating this information into either the matrix $\pmat$ or the distribution family $\pmatdistfamily$.

\paragraph{Applicability of the framework}
In summary,  our framework provides a lower bound on re-identification risk of data release that complements other privacy metrics that may be implemented by system designers, such as (differential) privacy budgets and aggregation (k-anonymity) guarantees. 

In cases where the `closed-world' assumption holds, we can use the methods developed in this work to provide even stronger upper bounds on re-identification risk. One potential such scenario of special interest is bounding \emph{insider risk}: the risk of an employee maliciously re-identifying user information across a company's systems. Given that such systems have strict controls on the flow of information (e.g. ACLs and auditing systems), we believe this framework can allow data protection officers to measure possibility of re-identification in such cases. 

Overall, we expect that the tools we have developed could be used by privacy advocates for a data-driven  verification of the privacy claims by data brokers and technology companies.

\appendix
\section{Proof of Theorem~\ref{thm:optimal_attacker}}
\begin{proof}
By Lemma~\ref{lemma:partial_info_bound} we know that, given a representation $o = (o_1, \ldots, o_r)$, the optimal attacker selects the user that maximizes 
\begin{equation}
    \max_i \mathbb{E}_\pmatdist\left[P[i, o] |W \right] =\max_i \prod \mathbb{E}_{\pmatdist_s}\left[P_s[i, o^s] |W_i^s \right],
    \label{eq:expprod}
\end{equation}

where the last equality follows by the independence of representations across time and users.
Let us calculate one of the terms in the above expression. Note that since $P_s[i,o^s]$ can only take two values $q_{\text{in}}$ and $q_{\text{out}}$, each factor in the above expectation is given by:
\begin{flalign}
    &\mathbb{E}_{\pmatdist_s}\left[P_s[i, o^s] |W_i^s \right] & \nonumber\\
    &=\pmatdist_s(P_s[i,o^s] = q_{\text{in}} | W_i^s) q_{\text{in}} + \pmatdist_s(P_s[i,o^s] =  q_{\text{out}} | W_i^s) q_{\text{out}} \nonumber \\
    & = q_{\text{out}} + \pmatdist_s(P_s[i,o^s] = q_{\text{in}}\ | \ W_i^s) (q_{\text{in}} - q_{\text{out}})
    \label{eq:expdecomp}
\end{flalign}

Note that if $W_i^s = o^s$, a straightforward application of Bayes rule and the fact that $\mathbb{P}(W_i^s = o^s | P_s[i,o^s] = q_{\text{in}}) = q_{\text{in}}$ yields
\begin{flalign*}
 \pmatdist_s(P_s[i,o^s] = q_{\text{in}}\ | \ W_i^s)  
&= \frac{q_{\text{in}}p_s[o^s]}{q_{\text{in}} p_s[o^s] + q_{\text{out}} ( 1 - p_s[o^s])}
\end{flalign*}
When $W_i^s \neq o^s$ we simply rewrite the conditional expectation by the more explainable $\mathbb{P}(o^s\in S_i^s | W_i^s \in S_i^s)$.
Plugging this expression in \eqref{eq:expdecomp} we see that
\begin{align*}
    E_{\pmatdist_s}\left[P_s[i, o^s] |W_i^s \right]
& =  
\left(q_{\text{out}} + \frac{(q_{\text{in}}-q_{\text{out}})q_{\text{in}} p_s[o^s]}{q_{\text{out}} + (q_{\text{in}} - q_{\text{out}})p_s[o^s]} \right) \ind_{W_i^s = o^s} \\ 
& + \left(q_{\text{out}} + (q_{\text{in}} - q_{\text{out}}) \mathbb{P}(o^s\in S_i^s | W_i^s \in S_i^s)\right)\ind_{W_i^s \neq o^s}
\end{align*}

The result follows by replacing this expression in \eqref{eq:expprod} and taking the logarithm of the product.
\end{proof}

\begin{proof}[Proof Of lemma \ref{lemma:estimation}]
Let $\mathbf{W}$ be a random variable sampled as follows. Sample a set of top topics $S \sim \pmatdist_s$. Then sample $\mathbf{W}$ according to 
\begin{equation*}
    \mathbb{P}(\mathbf{W} = o | S) = \left\{\begin{array}{cc}
    q_{\text{in}} & o \in S\\
    q_{\text{out}} & o \not\in S
    \end{array}
    \right.
\end{equation*}
It is very easy to see that the collection $(W_i^s)$ is an i.i.d. sample from random variable $\mathbf{W}$. Moreover we have that
\begin{align*}
\mathbb{P}(\mathbf{W} = o) &= \sum_{S\colon o \in S} \mathbb{P}(\mathbf{W} = o | S)\pmatdist_s(S)  
+ \sum_{S\colon o \notin S} \mathbb{P}(\mathbf{W} = o| S)\pmatdist_S(s)  \\
&=  q_{\text{in}}p_s[o] + (1 - p_s[o]) q_{\text{out}}
\end{align*}
Let $Y_i[o] = \frac{\ind_{W_i^s = o} - q_{\text{out}}}{q_{\text{in}} - q_{\text{out}}} $ and $\widehat{p}_s[o] = \frac{1}{n} \sum_{i=1}^n Y_i$. Using the fact that $|Y_i[o]| \leq \frac{1}{q_{\text{in}} - q_{\text{out}}}$ and $\mathbb{E}[Y_i[o]] = p_s[o]$, by Hoeffding's inequality we have with probability $1 - \frac{\delta}{N}$:
\begin{equation*}
|p_s[o] - \widehat{p}_s[o]| \leq \frac{1}{q_{\text{in}} - q_{\text{out}}} \sqrt{\frac{\log(2 N/\delta)}{2 n}}.
\end{equation*}
The result follows by using a union bound over all topics $o$.
\end{proof}
\section{Modeling assumption on topics}
In this section we discuss the modeling assumption on the Topics API used in the experiments (Section~\ref{sec:exp-attacks}). The results of this section show that we can model an attacker using only the parameters $p_s[o]$. 
\begin{assumption}
\label{assum:per_round_independence}
For a fixed round $s$ we assume that the random variable $S_i^s$ representing the top set of users satisfies
$
    \mathbb{P}(o \in S_i^s | o'\in S_i^s) = \alpha_s(o)
$
for all topics $o' \neq o$. 
\end{assumption}
The above assumption suggests some form of independence between the topics belonging to the top set. The following lemma shows that under this assumption $\alpha_s(o)$ is in fact a simple function of $p_s[o]$. 
\begin{lemma}
\label{lemma:alpha_calculation}
Let $S_i^s$ satisfy Assumption~\ref{assum:per_round_independence}. Then
\begin{equation*}
    \alpha_s(o) = \frac{4p_s[o]}{5 - p_s[o]}
\end{equation*}
\end{lemma}
\begin{proof}
By proposition \ref{prop:condsum} we know that:
\begin{align*}
    4p_s[o] &= \sum_{o' \neq o} \mathbb{P}((o, o') \in S_i^s)  
    = \sum_{o' \neq o} \mathbb{P}(o \in S_i^s |  o' \in S_i^s) P(o'\in S_i^s) \\
    &= \alpha_s(o)\sum_{o'\neq o}p_s[o'] 
    = \alpha_s(o)(5 - p_s[o]),
\end{align*}
where we have used Bayes rule and Assumption~\ref{assum:per_round_independence} for the second and third equalities respectively. The statement of lemma follows by rearranging terms.
\end{proof}

\begin{proposition}
\label{prop:condsum}
Let $o$ be a fixed topic. The following properties holds for any distribution over the top set $S_i^s$. 
\begin{equation}
    \label{eq:marginal}
    \sum_{o'\neq o}p_s[o'] = 5 - p_s[o]
\end{equation}
\begin{equation}
\label{eq:joint_sum}
    \sum_{o'\neq o}\mathbb{P}((o,o') \in S_i^s) = 4p_s[o]
\end{equation}

\end{proposition}
\begin{proof}
Fix a top set $S_i^s$ of 5 elements. It is then easy to see that
\begin{equation}
\label{eq:indic_marg}
    \ind_{o\in S_i^s} + \sum_{o'\neq o}\ind_{o\in S_i^s} = 5
\end{equation}
Similarly, we claim that 
\begin{equation}
\label{eq:indic_join}
    \sum_{o' \neq o}\ind_{o \in S_i^s} \ind_{o' \in {S_i^s}}= 4 \ind_{o \in S_i^s} 
\end{equation}
Indeed, if $o \notin S_i^s$ then the above expression is trivially true as both sides are $0$. If, on the other hand, $o\in S_i^s$, then there are exactly 4 other elements in the set so $\sum_{o'\neq o} \ind_{o' \in S_i^s} = 4$.  The result follows by taking expectation of \eqref{eq:indic_join} and \eqref{eq:joint_sum}.

\end{proof}

\bibliographystyle{plain}
\bibliography{references.bib}

\end{document}